\documentclass[10pt,twoside,a4paper]{article}

\usepackage{a4wide}
\usepackage{amssymb,amsmath,amsthm}
\usepackage{complexity}
\usepackage{enumitem}

\usepackage{authblk}

\usepackage[algosection,algoruled,vlined,nofillcomment]{algorithm2e}
\DontPrintSemicolon

\usepackage{tikz}
\usetikzlibrary{arrows,decorations,shapes,shadows,positioning,plotmarks,patterns,calc,automata}

\newcommand{\set}[1]{\ensuremath{\left\{#1 \right\}}}
\newcommand{\setcond}[2]{\set{#1 \, ; \: #2}}
\newcommand{\ceil}[1]{\ensuremath{\left\lceil #1 \right\rceil}}
\newcommand{\floor}[1]{\ensuremath{\left\lfloor #1 \right\rfloor}}
\newcommand{\ts}[0]{{t^\star}}
\newcommand{\ti}[0]{{\hat{t}}}

\newtheorem{theorem}{Theorem}[section]

\newtheorem{lemma}[theorem]{Lemma}

\title{A Greedy algorithm for local heating\footnotetext{This research is conducted within the project Network Optimization (17-10090Y) supported by Czech Science Foundation and the iCare project (11854) supported by STW. E-mail: \href{mailto:fink@ktiml.mff.cuni.cz}{fink@ktiml.mff.cuni.cz}, \href{mailto:j.l.hurink@utwente.nl}{j.l.hurink@utwente.nl}}}
\author[1,2]{Ji\v{r}\'i Fink}
\author[1]{Johann L. Hurink}
\affil[1]{University of Twente, Department of Computer Science, Mathematics and Electrical Engineering, P.O.~Box~217, 7500 AE, Enschede, The Netherlands}
\affil[2]{Department of Theoretical Computer Science and Mathematical Logic, Faculty of Mathematics and Physics, Charles University in Prague}
\date{}

\begin{document}
\maketitle
\begin{abstract}
This paper studies a planning problem for supplying hot water in domestic environment. Hereby, boilers (e.g. gas or electric boilers, heat pumps or microCHPs) are used to heat water and store it for domestic demands. We consider a simple boiler which is either turned on or turned off and is connected to a buffer of limited capacity. The energy needed to run the boiler has to be bought e.g. on a day-ahead market, so we are interested in a planning which minimizes the cost to supply the boiler with energy in order to fulfill the given heat demand. We present a greedy algorithm for this heating problem whose time complexity is ${\cal O}(T \alpha(T))$ where $T$ is the number of time intervals and $\alpha$ is the inverse of Ackermann function.
\end{abstract}

\section{Introduction}

In modern society, a significant amount of energy is consumed for heating water \cite{heating_water}. Almost every building is connected to a district heating system or equipped with appliances for heating water locally. Typical appliances for heating water are electrical and gas heating systems, heat pumps and Combined Heat and Power units (microCHP). The heated water is mostly stored in buffers to be prepared for the demands of the inhabitants of the building.

In this paper we consider a local heating systems which consist of
\begin{itemize} \itemsep=0em
	\item a supply which represents some source of energy (electricity, gas),
	\item a converter which converts the energy into heat (hot water),
	\item a buffer which stores the heat for later usage and
	\item a demand which represents the (predicted) consumption profile of heat.
\end{itemize}
A more formal definition of the considered setting for local heating and the used parameters and variables is given in Section \ref{sec:overview}. Although the presented model can consider arbitrary types of energy, in this paper we use \textit{electricity} and \textit{heat} to distinguish between consumed and produced energy. However note, that this simple model of a local heating system can not only be applied for heating water but has many other applications, e.g. heating demand of houses, fridges and freezers and inventory management. More details about those applications are given also in Section \ref{sec:overview}.

In the presented model we assume that the electricity used to heat the water has to be bought on a market. Although these prices are nowadays mostly fixed for private costumers, the supply companies are faced with variable prices resulting e.g. from a day ahead market. Furthermore, it is expected that in the future also the private customers get confronted with variable prices over time. This motivated the objective of minimizing the total cost of electricity consumed by the heating system during the planning period. Note that in cost or auction based control algorithms for Smart Grids, this objective is also used (see e.g. \cite{cost}).


\subsection{Problem statement and results} \label{sec:overview}

In the following we present a mathematical description of the studied model and a summary of the results of this paper.

The base of our modelling is a discretization of the planning horizon, meaning that we split the planning period into $T$ time intervals of the same length resulting in a set ${\cal T} = \set{1,\dots,T}$ of time intervals. In this paper, the letter $t$  is always used as an index of time intervals.

For the heating system, we consider a simple converter which has only two states: In every time interval the converter is either turned on or turned off for the complete time interval. The amount of produced heat during one time interval in which the converter is turned on is denoted by $H$. If the converter is turned off, then it consumes no electricity and produces no heat. Let $x_t \in \set{0,1}$ be the variable indicating whether the converter is running in time interval $t \in \cal{T}$ or not. Furthermore, if the converter is running, then it consumes some amount of electricity which costs $P_t$ in time interval $t \in \cal T$. In another words, $P_t$ is the price for running the converter in time interval $t \in \cal T$. Summarizing, the objective of the planning problem is minimizing the cost for producing the heat, which is given by $\sum_{t \in \cal{T}}{P_t x_t}$.

Coupled to the heating system is a buffer. The state of charge of the buffer in the beginning of time interval $t \in \cal{T}$ is denoted by $s_t$ and represents the amount of heat in the buffer. Note that $s_{T+1}$ is the state of charge at the end of planning period. Based of the physical properties of the buffer, the state of charge $s_t$ is limited by a lower bound $L_t$ and an upper bound $U_t$. In practice, these two bounds are usually constant over time since the upper bound $U_t$ is the capacity of buffer and the lower bound $L_t$ is zero. However, it may be useful to allow different values, e.g. a given initial state of charge can be modelled by setting $L_1$ and $U_1$ equal to the initial state. In this paper, we always assume that $L_1 = U_1$, meaning that the initial state of charge $s_1$ is fixed.

The (predicted) amount of consumed heat by the inhabitants of the house during time interval $t \in \cal{T}$ is denoted by $D_t$. This amount is assumed to be given and has to be supplied by either the buffer or the converter, or a combination of both, and is called the demand. In this paper, we study the off-line version of the problem, so we assume that both the demands $D_t$ and also the prices $P_t$ are given for the whole planning period already at the beginning of the planning period.

The variables $x_t$ specifying the operation of the converter and the states of charge of the buffer $s_t$ are restricted by the following constraints:
\begin{eqnarray}
	s_{t+1} = s_t + H x_t - D_t & \text{ for } & t \in {\cal T} \label{eq:charging} \\
	L_t \le s_t \le U_t & \text{ for } & t \in \set{1,\dots,T+1} \label{eq:buffer} \\
	x_t \in \set{0,1} & \text{ for } & t \in {\cal T}. \label{eq:converter}
\end{eqnarray}
Equation \eqref{eq:charging} is the charging equation of the buffer. During time interval $t \in \cal T$, the state of charge $s_t$ of the buffer is increased by the production of the converter which is $H x_t$ and it is decreased by demand $D_t$. Equations \eqref{eq:buffer} and \eqref{eq:converter} ensures that the domains of variables $s_t$ and $x_t$, respectively, are taken into account. As already mentioned, in this paper the objective function is to minimize the cost for the electricity needed to produce the heat  $\sum_{t \in \cal T} P_t x_t$.

In a previous paper \cite{cost_peak}, we presented an algorithm for the problem of minimizing cost for the local heating which is based on dynamic programming and it has the time complexity $\mathcal{O}(T^2)$. In Section \ref{sec:greedy} we prove that the optimal solution also can be calculated using a greedy algorithm. This greedy algorithm first sorts all time intervals by their prices $P_t$, and then it processes all time intervals one-by-one. In the basic version of the algorithm, the necessary updates in each step take time $\mathcal{O}(T)$, so the total time complexity of the algorithm is also $\mathcal{O}(T^2)$. In Section \ref{sec:union_find}, we then use the disjoint-set data structure of the union-find algorithm (see e.g. Cormen et al. \cite{Cormen}) to obtain a complexity of ${\cal O}(T \alpha(T))$ where $\alpha$ is the inverse of the Ackermann function. Hereby, we ignore the complexity of sorting the time intervals since the order may be a part of the input or be fount using a bucket sort algorithm (see e.g. Cormen et al. \cite{Cormen}).

\section{Related works and applications}

In the following we present related literature and give some possible further applications of this model.

Some related works can be found in the inventory management and lotsizing literature (see e.g. \cite{drexl1997lot,karimi2003capacitated} for reviews). In inventory control problems (see \cite{axsater2006inventory}) a buffer may represent an inventory of items, whereby a converter represent the production of items and the demand represent the ordered quantities. As our problem consists of only one commodity, the single item lot sizing problem is related (see \cite{brahimi2006single} for a review). Wagner and Whitin \cite{wagner1958dynamic} presented an $\mathcal{O}(T^2)$ algorithm for the uncapacitated lot-sizing problem which was improved by Federgruen and Tzur \cite{federgruen1991simple} to $\mathcal{O}(T \log T)$. On the other hand, Florian, Lenstra and Rinnooy \cite{florian1980deterministic} proved that the lot-sizing problem with upper bounds on production and order quantities is NP-complete. The computational complexity of the capacitated lot sizing problems is studied in \cite{bitran1982computational}. Pessoa at.al. \cite{de2019automatic} studied multiple variants of Multi-level capacitated lot-sizing problem which is an NP-hard problem, so they presented an automatic algorithm-generation approach based on heuristics and a multi-population genetic algorithm. Quezada et.al. \cite{quezada2019stochastic} proposed a stochastic dual dynamic integer programming algorithm for the multi-echelon multi-item lot-sizing problem. Our problem is a special case of the capacitated single item lot sizing problem which to our knowledge has not been considered in the literature.

One other related area is vehicle routing and scheduling (see e.g. \cite{laporte1992vehicle} for an overview of this area). For example, Lin, Gertsch and Russell \cite{hong2007linear} studied optimal vehicle refuelling policies. In their model, a refuelling station can provide an arbitrary amount of gas while our converter is restricted to two possible states of heat generation. Other papers on vehicle refuelling policies are more distant from our research since they consider that a car is routed on a graph (see e.g. \cite{sweda2012finding,lin2008finding}).

Where above we gave related but different optimization problems, in the following we present some possible applications of the model presented in this paper.

\begin{description}
	\item[Hot water:] Converter and buffer can be seen as a model of a simple electrical or gas boiler. Hereby, demand represents the consumption of hot water in a house.
	\item[House Heating:] The model may be used to express a very simple model for house heating. The converter represents a simple heater. The capacity of the buffer corresponds to thermal capacity of the heating system (e.g. hot water buffer or thermal capacity of concrete floors and walls) and the state of charge of the buffer is related to the temperature inside the house. Heat losses of the house may be modelled using the demand if we assume that the temperature difference inside the house does not have significant influence on the losses. More details about using thermal mass as a buffers is presented in \cite{Leeuwen2014thermal} and computing heat demands is explained in \cite{simple_control}.
	\item[Fridges and freezers:] A fridge essentially works in the opposite way than heating, so it may be modelled similarly. Hereby, the state of charge of the buffer again represents the temperature inside the fridge, but a higher state of charge means a lower temperature. The converter does not produce heat to the fridge but it decreases the temperature inside the fridge, so the converter increases the state of charge of the buffer (fridge). The demand decreases the state of charge of the fridge due to thermal loss and usage of the fridge by humans.
\end{description}

\section{Reformulation of the problem} \label{sec:preprocessing}

Where the problem formulation given in  Section \ref{sec:overview} is helpful to explain the problem, a reformulation of the problem presented in \cite{cost_peak} enables a better presentation and analysis of our algorithm. For sake of completeness, we give this reformulation in this section. We show that conditions \eqref{eq:charging} and \eqref{eq:buffer} can be replaced by one condition \eqref{eq:sum}.

First, we expand the recurrence formula \eqref{eq:charging} into an explicit equation
$$s_{t+1} = s_1 + \sum_{i=1}^t H x_i - \sum_{i=1}^t D_i.$$
Since we assume that the initial state of charge is given by $s_1 = L_1 = U_1$, we can replace $s_1$ by $L_1$ and substitute this into inequalities \eqref{eq:buffer}, leading to
$$\frac{L_{t+1} - L_1 + \sum_{i=1}^t D_i}{H} \le \sum_{i=1}^t x_i \le \frac{U_{t+1} - L_1 + \sum_{i=1}^t D_i}{H}.$$
Since the sum $\sum_{i=1}^t x_i$ is an integer between $0$ and $t$ we obtain the following simple constrains for this sums
\begin{eqnarray}
	A'_t \le \sum_{i=1}^t x_i \le B'_t \text{ for } t \in {\cal T} \label{eq:sum'}
\end{eqnarray}
where
\begin{eqnarray}
	A'_t = \max \set{0, \ceil{\frac{L_{t+1} - L_1 + \sum_{i=1}^t D_i}{H}}}
	\text{\hspace{3mm} and \hspace{3mm}}
B'_t = \min \set{t, \floor{\frac{U_{t+1} - L_1 + \sum_{i=1}^t D_i}{H}}}.
\label{eq:lower_upper}
\end{eqnarray}

Observe that the sequence of partial sums $\sum_{i=1}^t x_i$ for $t = 1,\ldots,T$ is non-decreasing and the difference of two consecutive partial sums is at most $1$. We say that a sequence $(Z_t)_t$ of $T+1$ integers $Z_0, Z_1, \dots, Z_T$ satisfies \eqref{eq:seq} if
\begin{equation}\label{eq:seq}
\begin{aligned}
	& Z_0 = 0, \\
	& Z_{t-1} \le Z_t \le Z_{t-1} + 1 \text{ for all } t \in {\cal T}.
\end{aligned}
\end{equation}
From parameters $A'_t$ and $B'_t$ we can be easily compute parameters $A_t$ and $B_t$ such that sequences $(A_t)_t$ and $(B_t)_t$ satisfy \eqref{eq:seq} and the binary variables $x_t$ satisfy \eqref{eq:sum'} if and only if they satisfy
\begin{eqnarray}
	A_t \le \sum_{i=1}^t x_i \le B_t \text{ for } t \in {\cal T}. \label{eq:sum}
\end{eqnarray}
For more details, see \cite{cost_peak}.

\section{Greedy algorithm}\label{sec:greedy}

In this section we present a greedy algorithm for the problem of fulfilling the heat demand with minimal cost which is based on the following formulation of the problem:
\begin{eqnarray}
\text{Minimize} & \sum_{t \in \cal T} P_t x_t \nonumber \\
\text{such that} & A_t \le \sum_{i=1}^t x_i \le B_t & \text{ for } t \in {\cal T} \label{eq:single}\\
& x_t \in \set{0,1} & \text{ for } t \in {\cal T} \nonumber
\end{eqnarray}
For the following, we assume that the given bounds $A_t$ and $B_t$ are already such that the sequences $(A_t)_t$ and $(B_t)_t$ satisfies \eqref{eq:seq}.

The first natural question is under which conditions problem \eqref{eq:single} has a feasible solution. An obvious condition for the existence of a feasible solution is that $A_t \le B_t$ for every $t \in \cal T$. This condition is also sufficient, since in this case $x_t = A_t - A_{t-1}$ for $t \in \cal T$ gives a feasible solution. Summarizing, we get the following lemma.

\begin{lemma}\label{lem:feasible}
The problem \eqref{eq:single} has a feasible solution if and only if
\begin{equation}\label{eq:feasible}
A_t \le B_t \text{ for every } t \in \cal T.
\end{equation}
\end{lemma}

Since the condition \eqref{eq:feasible} can easily be evaluated in linear time, we assume in the remainder of the paper that the problem \eqref{eq:single} has a feasible solution. To solve the problem, we use the classical greedy approach. First, the time intervals are sorted by prices $P_t$. Then, the time intervals are processed in order of increasing prices and the converter is turned on in time interval $\ts \in \cal T$, if there exists a feasible solution with $x_\ts = 1$. Note, that the existence of such a feasible solution implies
$$A_{\ts-1} \le \sum_{i=1}^{\ts-1} x_i < \sum_{i=1}^{\ts} x_i \le B_\ts$$
which leads to the following lemma.

\begin{lemma}\label{lem:x1}
If the problem \eqref{eq:single} has a feasible solution $(x_t)_t$ satisfying $x_\ts = 1$ for a given $\ts \in \cal T$, then the inequality $A_{\ts-1} < B_\ts$ must hold.
\end{lemma}


Actually, the condition $A_{\ts-1} < B_\ts$ is sufficient and the proof follows from Lemmas \ref{lem:update} and \ref{lem:optimal-base}.

The greedy algorithm starts with the (infeasible) solution $x_t = 0$ for every $t \in T$. Then it finds the cheapest time interval $\ts$ satisfying $A_{\ts-1} < B_\ts$ and it sets $x_\ts := 1$. Next, the values of sequences $(A_t)_t$ and $(B_t)_t$ have to be adopted to incorporate the choice $x_\ts = 1$. The following lemma gives update rules for the values of $(A_t)_t$ and $(B_t)_t$ in every step and shows that this update is correct.


\begin{lemma}\label{lem:update}
Let $\ts \in \cal T$ be an interval satisfying $A_{\ts-1} < B_{\ts}$, and let
\begin{eqnarray*}
t_A = \max\setcond{t \in {\cal T}}{A_t = A_{\ts-1}} + 1 \hspace{10mm} & t_B = \min\setcond{t \in {\cal T}}{B_t = B_\ts} \\
A^\star_t =
\begin{cases}
A_t & \text{ if } t < t_A \\
A_t - 1 & \text{ if } t \ge t_A \\
\end{cases}
\hspace{20mm} & B^\star_t =
\begin{cases}
B_t & \text{ if } t < t_B \\
B_t - 1 & \text{ if } t \ge t_B. \\
\end{cases}
\end{eqnarray*}
Then also the sequences $(A^\star_t)_t$ and $(B^\star_t)_t$ satisfy \eqref{eq:seq} and for every 0-1 sequence $(x_t)_t$ with $x_\ts = 1$ the condition \eqref{eq:sum} holds if and only if
\begin{equation}\label{eq:star}
A_t^\star \le \sum_{\substack{i=1 \\ i \neq \ts}}^t x_i \le B_t^\star \text{ for every } t \in \cal T.
\end{equation}
\end{lemma}


\begin{proof}
Since the sequences $(A^\star_t)_t$ and $(B^\star_t)_t$ differ from sequences $(A_t)_t$ and $(B_t)_t$ by removing one step of the step function at time intervals $t_A$ and $t_B$, respectively, $(A^\star_t)_t$ and $(B^\star_t)_t$ satisfy \eqref{eq:seq}. In order to prove the second part, let $(x_t)_t$ be a 0-1 sequence with $x_\ts = 1$. For such a sequence, the condition \eqref{eq:sum} is equivalent to
\begin{equation}\label{eq:equivalent}
\begin{aligned}
& A_t \le \sum_{\substack{i=1 \\ i \neq \ts}}^t x_i \le B_t && \;\text{ for every } t < \ts \text{ and} \\
& A_t - 1 \le \sum_{\substack{i=1 \\ i \neq \ts}}^t x_i \le B_t - 1 && \;\text{ for every } t \ge \ts.
\end{aligned}
\end{equation}
Thus, it remains to prove that conditions \eqref{eq:star} and \eqref{eq:equivalent} are equivalent. We show this only for the lower bounds, since the upper bound case is similar. Observe that the lower bounds of \eqref{eq:star} and \eqref{eq:equivalent} only differ for time intervals $t \in \cal T$ with $\ts \le t < t_A$. For such $t$ it holds that $A_t^\star = A_t$ and, thus, we only have to prove that \eqref{eq:equivalent} implies \eqref{eq:star} since the lower bound in \eqref{eq:star} is stronger. However, the implication follows directly from
$$A^\star_t = A_t = A_{\ts-1} \le \sum_{i=1}^{\ts-1} x_i \le \sum_{\substack{i=1 \\ i \neq \ts}}^t x_i.$$
\end{proof}

In practical applications, the price of electricity is usually positive. However, the presented greedy algorithm works also if the price $P_t$ is negative for some $t \in \cal T$. If all prices are non-negative, then without loss of generality we can assume that $A_T = B_T$ since there is an optimal solution which turns the converter on only $A_T$-times (that is, there exists an optimal solution with $\sum_{t \in \cal T} x_t = A_T$). In the general case where prices can be negative, the value of the objective function may be improved by turning the converter on more often. In the latter case, we need to extend the condition $A_{\ts-1} < B_\ts$ of the greedy algorithm to a condition which also considers negative prices. The new condition is
\begin{equation}\label{eq:apply}
A_{\ts-1} < B_\ts \text{ and } (A_{\ts-1} < A_T \text{ or } P_\ts < 0).
\end{equation}
If $A_{\ts-1} = A_T$ then the total minimal number of runs of the converter has to be reached already before the time interval $\ts$. Thus, the lower bound $(A_t)_t$ does not force the converter to be on in the time interval $\ts$. In this case, it is obvious that an optimal solution satisfies $x_\ts = 0$ unless the price $P_\ts$ is negative.

\begin{algorithm}[t]
\caption{Greedy algorithm for minimizing cost. \label{alg:basic}}
\KwIn{Sequences $(A_t)_t$ and $(B_t)_t$ satisfying \eqref{eq:seq} and \eqref{eq:feasible}}
\KwOut{Optimal solution $(x_t)_t$ the problem \eqref{eq:single}}
\BlankLine
initialization: $x_t := 0$ for all $t \in \cal T$\;
\For{$\ts \in \cal T$ sorted by prices $(P_t)_t$}{
\If{the condition \eqref{eq:apply} is satisfied}{
$x_\ts := 1$\;
Apply Lemma \ref{lem:update}
}}
\BlankLine
\KwRet{Optimal solution $(x_t)_t$}
\end{algorithm}

The greedy algorithm is summarized in Algorithm \ref{alg:basic}. In the following, mathematical induction is used to prove that this greedy algorithm finds an optimal solution. The following two lemmas provide the base of the induction and the induction step.

\begin{lemma}\label{lem:optimal-base}
If $A_T = 0$ and there is no $\ts \in \cal T$ such that $B_\ts > 0$ and $P_\ts < 0$, then $x_t = 0$ for all $t \in \cal T$ is an optimal solution.
\end{lemma}
\begin{proof}
Since $A_T = 0$ it follows that $A_t = 0$ for all $t \in \cal T$ and thus, the trivial solution $x_t = 0$ for all $t \in \cal T$ is feasible. Let $\bar{t} = \max\setcond{x \in {\cal T}}{B_t = 0}$ and let $(\bar{x}_t)_t$ be an arbitrary feasible solution. Observe that $\bar{x}_t = 0$ for $t \le \bar{t}$ and $P_t \ge 0$ for $t > \bar{t}$. Hence,
$$\sum_{t \in \cal T}P_t \bar{x}_t = \sum_{t > \bar{t}}P_t \bar{x}_t \ge 0 = \sum_{t \in \cal T}P_t x_t$$
which implies that the solution $(x_t)_t$ is optimal.
\end{proof}


\begin{lemma}\label{lem:optimal-step}
Assuming that there exists a time interval $\ts \in \cal T$ satisfying \eqref{eq:apply}, let $\ts$ be the time interval satisfying \eqref{eq:apply} with the minimal price $P_\ts$. Then, there exists an optimal solution $(x_t)_t$ such that $x_\ts = 1$.
\end{lemma}
\begin{proof}
We prove the lemma indirectly by proving that for every feasible solution $(\bar{x}_t)_t$ there exists a feasible solution $(\hat{x}_t)_t$ such that $\hat{x}_\ts = 1$ and $\sum_{t \in \cal T}{P_t \hat{x}_t} \le \sum_{t \in \cal T}{P_t \bar{x}_t}$. Since we assume that there always exists a feasible solution, the lemma follows from this observation.

Let $(\bar{x}_t)_t$ be a feasible solution. If $\bar{x}_t = 1$, we are done. Thus, we assume that $\bar{x}_t = 0$, and we consider two cases.

\textit{Case 1} Assume $\sum_{i=1}^{\ts-1}{\bar{x}_i} > A_{\ts-1}$: Let $\ti = \max\setcond{t < \ts}{\bar{x}_t = 1}$ which is well-defined since $\sum_{i=1}^{\ts-1}{\bar{x}_i} > A_{\ts-1} \ge 0$. The new solution now is defined by $\hat{x}_\ti = 0$ and $\hat{x}_\ts = 1$ and $\hat{x}_t = \bar{x}_t$ for $t \in {\cal T} \setminus \set{\ti,\ts}$. In order to prove that $(\hat{x}_t)_t$ fulfills the mentioned conditions, we first prove that it is feasible. The equality $\sum_{i = 1}^t \hat{x}_t = \sum_{i = 1}^t \bar{x}_t$ does not hold only for time intervals $t$ with $\ti \le t < \ts$. However, for such $t$ it holds that
$$A_t \le A_{\ts-1} \le \sum_{i=1}^{\ts-1}{\bar{x}_i} - 1 = \sum_{i=1}^{t}{\bar{x}_i} - 1 = \sum_{i=1}^{t}{\hat{x}_i} < \sum_{i=1}^{t}{\bar{x}_i} \le B_t.$$
Hence, $(\hat{x}_t)_t$ is feasible.

Next, for sake of contradiction we assume $\sum_{t \in \cal T} P_t \hat{x}_t > \sum_{t \in \cal T} P_t \bar{x}_t$, implying that $P_\ti < P_\ts$. If $\ti$ satisfies \eqref{eq:apply}, we have a contradiction with the definition of $\ts$, so $\ti$ does not satisfy \eqref{eq:apply}. Applying Lemma \ref{lem:x1} with $\bar{x}_\ti = 1$ we get that $A_{\ti-1} < B_\ti$ which implies that $A_{\ti-1} = A_T$ and $P_\ti \ge 0$. Since $\ti < \ts$ we have $A_{\ti-1} = A_{\ts-1} = A_T$. Furthermore, we have $P_\ts > P_\ti \ge 0$, meaning that $\ts$ does not satisfy \eqref{eq:apply} which is a contradiction. Thus, $\sum_{t \in \cal T} P_t \hat{x}_t \le \sum_{t \in \cal T} P_t \bar{x}_t$.

\textit{Case 2} Assume $\sum_{i=1}^{\ts-1}{\bar{x}_i} = A_{\ts-1}$: For this case we have to consider two subcases depending on whether there exists a $t \ge \ts$ such that $\bar{x}_t = 1$. For both subcases we only give the corresponding solution $(\hat{x}_t)_t$; the proof of feasibility and cost is similar to Case 1.

\textit{Case 2.1} Assume $\bar{x}_t = 0$ for all $t \ge \ts$: The new solution is defined by $\hat{x}_\ts = 1$ and $\hat{x}_t = \bar{x}_t$ for $t \in {\cal T} \setminus \set{\ts}$.

\textit{Case 2.2} Assume that $\ti = \min\setcond{t \ge \ts}{\bar{x}_t = 1} \in \cal T$ is well-defined: The new solution is defined by $\hat{x}_\ti = 0$ and $\hat{x}_\ts = 1$ and $\hat{x}_t = \bar{x}_t$ for $t \in {\cal T} \setminus \set{\ti,\ts}$.

\end{proof}

\begin{theorem}\label{thm:greedy}
Algorithm \ref{alg:basic} finds an optimal solution for the local heating problem in time $\mathcal{O}(T^2)$.
\end{theorem}
\begin{proof}
Algorithm \ref{alg:basic} assumes that the pre-computed sequences $(A_t)_t$ and $(B_t)_t$ satisfy \eqref{eq:seq} and \eqref{eq:feasible}. The initialization can easily be computed in linear time.

We use induction on the number of updates according to Lemma \ref{lem:update} to prove that Algorithm \ref{alg:basic} finds an optimal solution. As base of the induction, we assume that Algorithm \ref{alg:basic} never applies Lemma \ref{lem:update}. In this case, there is no time interval $\ts$ which satisfies \eqref{eq:apply}. Hence, $A_T = 0$ and Lemma \ref{lem:optimal-base} implies that the trivial solution $x_t = 0$ for all time intervals $t \in \cal T$ is an optimal solution.

For the induction step, assume that Algorithm \ref{alg:basic} finds a time interval $\ts$ satisfying \eqref{eq:apply} with the minimal price $P_\ts$ and that Lemma \ref{lem:update} is applied with this $\ts$. By the induction hypothesis, Algorithm \ref{alg:basic} finds an optimal solution $(x_t)_t$ for the instance with sequences $(A^\star_t)$ and $(B^\star_t)_t$ where $t \in {\cal T} \setminus \set{\ts}$. Algorithm \ref{alg:basic} now extends this solution by setting  $x_\ts = 1$. This solution $(x_t)_t$ is feasible for $(A_t)_t$ and $(B_t)_t$ by Lemma \ref{lem:update}.

By Lemma \ref{lem:optimal-step} there exists an optimal solution $(\bar{x}_t)_t$ satisfying $\bar{x}_\ts = 1$. By Lemma \ref{lem:update} the solution $(\bar{x}_t)_t$ is feasible for the instance with sequences $(A^\star_t)$ and $(B^\star_t)_t$ where $t \in {\cal T} \setminus \set{\ts}$. From induction hypothesis it follows that $\sum_{t \in {\cal T} \setminus \set{\ts}} P_t x_t \le \sum_{t \in {\cal T} \setminus \set{\ts}} P_t \bar{x}_t$. Hence, $\sum_{t \in {\cal T}} P_t x_t \le \sum_{t \in {\cal T}} P_t \bar{x}_t$ which implies that $(x_t)_t$ is an optimal solution.

Since Lemma \ref{lem:update} is called at most $T$-times and every step is evaluated in time $\mathcal{O}(T)$, the total time complexity is $\mathcal{O}(T^2)$.
\end{proof}

Algorithm \ref{alg:basic} has quadratic complexity because the updates of sequences $(A_t)_t$ and $(B_t)_t$ take linear time. This complexity easily can be improved using a binary tree. The basic idea is that values of $(A_t)_t$ and $(B_t)_t$ are handled independently by two separate balanced binary trees (see e.g. Cormen et al. \cite{Cormen}). In the following, we only describe the tree for the sequence $(A_t)_t$ since the tree for sequence $(B_t)_t$ can be handled analogously. Let $A^d_t = A_t - A_{t-1}$. The leaves of the tree store the value $A^d_t$ for all the time intervals $t \in \cal T$. Time intervals are assigned to leaves in a sorted way where the left subtree of every inner vertex contains earlier time intervals than the right subtree. Every inner vertex of the tree stores the sum of values of $A^d_t$ for all leaves $t$ in the subtree. Since the tree is constructed to be balanced, the length of every path from the root to a leaf is $\log_2(T) + \mathcal{O}(1)$. This binary tree is constructed in time $\mathcal{O}(T)$. It is now straight-forward to determine the simple exercise to find out how values $A_t$ and $B_t$ are determined and how to update both trees when Lemma \ref{lem:update} is applied. Both operations are performed in logarithmic time, so these binary trees improve the time complexity of Algorithm \ref{alg:basic} to $\mathcal{O}(T \log T)$. We skip more details because in the next section an even faster data structure is presented.

\section{Union-find} \label{sec:union_find}

In this section we use the disjoint-set data structure of the union-find algorithm (see e.g. Cormen et al. \cite{Cormen}) to store and update values of sequences $(A_t)_t$ and $(B_t)_t$ to reduce the time complexity of the presented algorithm.

A disjoint-set data structure is a data structure that keeps track of a set of elements partitioned into a number of disjoint (non-overlapping) subsets. A union-find algorithm is an algorithm that performs two useful operations on such a data structure:
\begin{description}
	\item[Find:] Determine which subset a particular element is in. This can be used for determining if two elements are in the same subset.
	\item[Union:] Join two subsets into a single subset.
\end{description}
Using a technique called path compression, both operation have amortized complexity $\mathcal{O}(\alpha(n))$ where $\alpha$ is the inverse of Ackermann function and $n$ is the number of elements.

Let us first summarize the operations required by Algorithm \ref{alg:basic}: We need to determine whether the condition \eqref{eq:apply} is satisfied and apply the updates in Lemma \ref{lem:update}. More specifically,
we need a data structure supporting the following operations for a given time interval $\ts$.

\begin{enumerate}[label={(U\arabic{enumi})}]
\item Find time intervals $t_A$ and $t_B$ as defined in Lemma \ref{lem:update}. \label{it:tAtB}
\item Determine whether $A_{\ts-1} < A_T$. \label{it:Amaximal}
\item Determine whether $A_{\ts-1} < B_\ts$. \label{it:AtBt}
\item Modify the data structure so that it gives the correct response for the above queries after the updates defined in Lemma \ref{lem:update}. \label{it:update}
\end{enumerate}

Consider a partitioning of time intervals $\ts \in \cal T$ by values $t_B$. The update \ref{it:update} can be performed by uniting two consecutive partitions and setting the value $t_B$ to be the smallest of the original partitions. Therefore, it is natural to use the disjoint-set data structure to store the partitioning time intervals $\ts \in \cal T$ by the values $t_B$ for every partition. Similarly, we use another disjoint-set data structures to partition the set of time intervals $\cal T$ according to the values $t_A$. So, this data structure is able to answer the queries \ref{it:tAtB} and \ref{it:Amaximal} and update itself \ref{it:update}. Note that we do not use these two disjoint-set data structures to store values $A_{\ts-1}$ and $B_\ts$ since their update may be too slow. However, we are able to determine whether $B_{t_1} = B_{t_2}$ for time intervals $t_1, t_2 \in \cal T$ (and similarly $A_{t_1-1} = A_{t_2-1}$) since $B_{t_1} = B_{t_2}$ if and only if $t_1$ and $t_2$ belong to the same partition in the disjoint-set data structure. This fact is used later to determine whether \eqref{eq:store} is satisfied.


In order to simplify further notation, let $\sim$ be a relation on the set of time intervals $\cal T$ such that $t_1 \sim t_2$ if $A_{t_1-1} = A_{t_2-1}$ and $B_{t_1} = B_{t_2}$ where $t_1, t_2 \in \cal T$. Observe that $\sim$ is an equivalence relation on $\cal T$ in which every factor class contains a set of consecutive time intervals. Factor classes of the equivalence $\sim$ are called B-A-sets. Since values of $A_t$ for all time intervals $t$ of one B-A-set $S$ are equal, we denote this value by $A_S$. Similarly, $B_S$ denotes the $B_t$ value of all time intervals $t$ of a B-A-set $S$. Update \ref{it:update} modifies the relation $\sim$, but the only change in the relation $\sim$ is that some B-A-sets are united. In fact, one update \ref{it:update} leads to at most two unions: B-A-sets containing time intervals $t_A-1$ and $t_A$ may be united and B-A-sets containing time intervals $t_B-1$ and $t_B$ may be united. Hence, we use the third disjoint-set data structure to store the partitioning into B-A-sets.

Let the difference $B_{\ts} - A_{\ts-1}$ denote $D_\ts$. Since all time intervals $\ts$ in one B-A-set have the same value of the difference $D_\ts$, this difference $D_\ts$ can be stored in every B-A-set. However, this approach is insufficient to reach the desired time complexity. The difference $D_\ts$ is actually stored only in some selected B-A-sets that form some kind of local minima of the sequence $(D_t)_t$ since \ref{it:AtBt} only requires to determine whether $D_\ts$ is zero or positive. Let us consider one B-A-set $S$ and let $S^-$ and $S^+$ be the preceding and the succeeding B-A-set, respectively. Observe that if $A_{S^-} = A_S$, then $B_{S^-}+1 = B_S$ and therefore $D_S = D_{S^-} + 1$. Similarly, if $B_S = B_{S^+}$, then $A_S+1 = A_{S^+}$ and therefore $D_S = D_{S^+} + 1$. In both cases it is not necessary to store the value $D_S$ since the facts that $D_{S^-} \ge 0$ or $D_{S^+} \ge 0$ imply that $D_S > 0$. The difference $D_S$ is stored in a B-A-set $S$ if and only if
\begin{equation}\label{eq:store}
	A_{S^-} < A_S \text{ and } B_S < B_{S^+}.
\end{equation}
Note that we are able to determine whether \eqref{eq:store} is satisfied using the first two union-set data structures. In summary, for $\ts \in S$ it holds that $A_{\ts-1} = B_\ts$ if and only if $D_S$ is stored in the B-A-set $S$ of the third disjoint-set data structure and $D_S = 0$.

In order to perform the update \ref{it:update}, let $S_1, S_2, S_3$, $S_4$ and $S^\star$ be B-A-sets containing time intervals $t_B-1$, $t_B$, $t_A-1$, $t_A$ and $\ts$; respectively. Note that $S_2$, $S_3$ and $S^\star$ may be the same. Observe that the difference $D_t$ is changed only for time intervals $t$ with $t_A \le t < t_B$, where $D_t$ decreases by one. Therefore, the evaluation of the condition \eqref{eq:store} may change only for sets $S_1, S_2, S_3$ and $S_4$. Furthermore, $S^\star$ is the only B-A-set which can satisfy \eqref{eq:store} and for which the stored difference $D_{S^\star}$ can change. The only possible changes in the partitioning of the time intervals $\cal T$ into B-A-sets are uniting $S_1$ and $S_2$ and uniting $S_3$ and $S_4$.

In the following, we only discuss updates of B-A-sets $S_1$ and $S_2$ since updates of B-A-sets $S_3$ and $S_4$ are analogous. Note that $B_{S_1}+1 = B_{S_2}$ and new values according to Lemma \ref{lem:update} satisfy $B^\star_{S_1} = B^\star_{S_2} = B_{S_1}$. If $A_{S_1} < A_{S_2}$, then B-A-sets $S_1$ and $S_2$ are not united and the value of the difference $D_{S_1}$ is deleted (if it already has been stored). If $A_{S_1} = A_{S_2}$, then B-A-sets $S_1$ and $S_2$ are united. Observe that if the united set satisfies \eqref{eq:store}, then $S_1$ has satisfied \eqref{eq:store} and the difference of the united set is the difference of $S_1$. Furthermore, if $S^\star \neq S_2$ and $S^\star \neq S_3$, then $S^*$ satisfies \eqref{eq:store} and the value $D_{S^\star}$ decreases by 1. All updates of the disjoint-set data structure of B-A-sets are summarized in Algorithm \ref{alg:update}.

\begin{algorithm}[t]
\caption{Update of the disjoint-set data structure for B-A-sets \label{alg:update}}
\KwIn{Time interval $\ts$}
\BlankLine
Find time intervals $t_A$ and $t_B$\;
Find B-A-sets $S_1, S_2, S_3$, $S_4$ and $S^\star$ containing $t_B-1$, $t_B$, $t_A-1$, $t_A$ and $\ts$\, respectively\;
\If{$A_{S_1} = A_{S_2}$}{
	Union of sets $S_1$ and $S_2$ into a set $S_{12}$\;
	\If{$S_{12}$ satisfies \eqref{eq:store}}{
		Set the difference $D_{S_{12}}$ to be the difference $D_{S_1}$ before the last uniting\;}}
\If{$B_{S_3} = B_{S_4}$}{
	Union of sets $S_3$ and $S_4$ into sets $S_{34}$\;
	\If{$S_{34}$ satisfies \eqref{eq:store}}{
		Set the difference $D_{S_{34}}$ to be the difference $D_{S_4}$ before the last uniting\;}}
\If{$S^\star \not= S_2$ and $S^\star \not= S_3$}{
	Decrease the difference $D_{S^\star}$ by one\;}
\end{algorithm}

Since the number of operations find and union on the disjoint-set data structures is $O(T)$ and the amortized complexity of these operations is $\mathcal{O}(\alpha(T))$, the following theorem follows.

\begin{theorem}
Algorithm \ref{alg:basic} with the disjoint-set data structures finds an optimal solution for the local heating problem in time $\mathcal{O}(T \alpha(T))$.
\end{theorem}

\section{Conclusion}

This paper presents a $\mathcal{O}(T\, \alpha(T))$ algorithm for local heating problem, which is based on a greedy algorithm and where the low complexity of the algorithm results from the use of a sophisticated data structure.

Looking at the settings for heating systems in practice, we note that for a part of the systems a valve can be used to control the heat flow to a buffer (e.g. district heating) and thereby the decision set has a continuous domain. This mathematically means that the constrain \eqref{eq:converter} is replaced by $0 \le x_t \le 1$. We believe that it is possible to adopt our algorithm to this case, although some parts may become more technical.

\bibliographystyle{plain}
\bibliography{sg}

\begin{thebibliography}{10}

\bibitem{heating_water}
C.~Aguilar, D.~J. White, and D.~L. Ryan.
\newblock Domestic water heating and water heater energy, consumption in
  canada.
\newblock {\em The Canadian Building Energy End-use Data and Analysis Centre},
  2, 2005.

\bibitem{bitran1982computational}
G.~R. Bitran and H.~H. Yanasse.
\newblock Computational complexity of the capacitated lot size problem.
\newblock {\em Management Science}, 28(10):1174--1186, 1982.

\bibitem{brahimi2006single}
N.~Brahimi, S.~Dauzere-Peres, N.~M. Najid, and A.~Nordli.
\newblock Single item lot sizing problems.
\newblock {\em European Journal of Operational Research}, 168(1):1--16, 2006.

\bibitem{Cormen}
T.~H. Cormen, C.~E. Leiserson, R.~Rivest, and C.~Stein.
\newblock {\em Introduction to Algorithms}.
\newblock MIT Press, 2001.

\bibitem{de2019automatic}
Lu{\'\i}s~Filipe de~Ara{\'u}jo~Pessoa, Bernd Hellingrath, and Fernando~Buarque
  de~Lima~Neto.
\newblock Automatic generation of optimization algorithms for production
  lot-sizing problems.
\newblock In {\em 2019 IEEE Congress on Evolutionary Computation (CEC)}, pages
  1774--1781. IEEE, 2019.

\bibitem{drexl1997lot}
A.~Drexl and A.~Kimms.
\newblock Lot sizing and scheduling-survey and extensions.
\newblock {\em European Journal of Operational Research}, 99(2):221--235, 1997.

\bibitem{federgruen1991simple}
A.~Federgruen and M.~Tzur.
\newblock A simple forward algorithm to solve general dynamic lot sizing models
  with n periods in 0 (n log n) or 0 (n) time.
\newblock {\em Management Science}, 37(8):909--925, 1991.

\bibitem{cost_peak}
J.~Fink and J.L. Hurink.
\newblock Minimizing costs is easier than minimizing peaks when supplying the
  heat demand of a group of houses.
\newblock {\em European Journal of Operational Research}, 242:644--650, 2015.

\bibitem{simple_control}
J.~Fink, R.~P. van Leeuwen, J.~L. Hurink, and G.~J.~M. Smit.
\newblock Linear programming control of a group of heat pumps.
\newblock In {\em Energy, Sustainability and Society, 5:33.}, 2015.

\bibitem{florian1980deterministic}
M.~Florian, J.~K. Lenstra, and A.~H.~G. Rinnooy~Kan.
\newblock Deterministic production planning: Algorithms and complexity.
\newblock {\em Management science}, 26(7):669--679, 1980.

\bibitem{karimi2003capacitated}
B.~Karimi, S.~M.~T. Fatemi~Ghomi, and J.~M. Wilson.
\newblock The capacitated lot sizing problem: a review of models and
  algorithms.
\newblock {\em Omega}, 31(5):365--378, 2003.

\bibitem{laporte1992vehicle}
G.~Laporte.
\newblock The vehicle routing problem: An overview of exact and approximate
  algorithms.
\newblock {\em European Journal of Operational Research}, 59(3):345--358, 1992.

\bibitem{lin2008finding}
S.-H. Lin.
\newblock Finding optimal refueling policies: a dynamic programming approach.
\newblock {\em Journal of Computing Sciences in Colleges}, 23(6):272--279,
  2008.

\bibitem{hong2007linear}
S.-H. Lin, N.~Gertsch, and J.~R. Russell.
\newblock A linear-time algorithm for finding optimal vehicle refueling
  policies.
\newblock {\em Operations Research Letters}, 35(3):290--296, 2007.

\bibitem{cost}
A.~Molderink, V.~Bakker, M.~G.~C. Bosman, J.~L. Hurink, and G.~J.~M. Smit.
\newblock Management and control of domestic smart grid technology.
\newblock {\em IEEE Transactions on Smart Grid}, 1(2):109--119, 2010.

\bibitem{quezada2019stochastic}
Franco Quezada, C{\'e}line Gicquel, and Safia Kedad-Sidhoum.
\newblock Stochastic dual dynamic integer programming for a multi-echelon
  lot-sizing problem with remanufacturing and lost sales.
\newblock In {\em 2019 6th International Conference on Control, Decision and
  Information Technologies (CoDIT)}, pages 1254--1259. IEEE, 2019.

\bibitem{axsater2006inventory}
A.~Sven.
\newblock {\em Inventory control}, volume~90 of {\em International Series in
  Operations Research and Management Science}.
\newblock Springer, 2006.

\bibitem{sweda2012finding}
T.~M. Sweda and D.~Klabjan.
\newblock Finding minimum-cost paths for electric vehicles.
\newblock In {\em Electric Vehicle Conference (IEVC), 2012 IEEE International},
  pages 1--4, 2012.

\bibitem{Leeuwen2014thermal}
R.~P. van Leeuwen, J.~Fink, J.~B. de~Wit, and G.~J.~M. Smit.
\newblock Thermal storage in a heat pump heated living room floor for urban
  district power balancing, effects on thermal comfort, energy loss and costs
  for residents.
\newblock In {\em Smartgreens 2014}, 2014.

\bibitem{wagner1958dynamic}
H.~M. Wagner and T.~M. Whitin.
\newblock Dynamic version of the economic lot size model.
\newblock {\em Management science}, 5(1):89--96, 1958.

\end{thebibliography}

\end{document}